\documentclass[10pt,conference]{IEEEtran}
\IEEEoverridecommandlockouts
\usepackage[utf8]{inputenc}

\usepackage{braket}
\usepackage{amsfonts,amssymb}
\usepackage{graphicx}
\usepackage{subfigure}
\usepackage{color}
\usepackage{tikz}
\usetikzlibrary{quantikz}
\usepackage{amsmath,amsthm}
\usepackage{url}
\usepackage{textcomp}

\newtheorem{thm}{Theorem}

\newtheorem{lem}{Lemma}

\newtheorem{exam}{Example}

\usepackage[normalem]{ulem}

\newcommand{\supp}[1]{{\text{supp}({#1})}}

\begin{document}

\title{Decision Diagrams for Symbolic Verification of Quantum Circuits}

% \author[1]{Xin~Hong}
% \author[4]{Wei-Jia Huang}
% \author[5]{Wei-Chen Chien}
% \author[1]{Yuan Feng}
% \author[4]{Min-Hsiu Hsieh}
% \author[1]{Sanjiang Li}
% \author[5]{Chia-Shun Yeh}
% \author[1,2,3*]{Mingsheng~Ying\thanks{*Corresponding authors:\{mingsheng.ying, yuan.feng, sanjiang.li\}@uts.edu.au.}}
% \affil[1]{Centre for Quantum Software and Information, University of Technology Sydney, Australia}
% \affil[2]{State Key Laboratory of Computer Science, Institute of Software, Chinese Academy of Sciences, China}
% \affil[3]{Department of Computer Science and Technology, Tsinghua University, China}
% \affil[4]{Hon Hai Quantum Computing Research Center, Taipei, Taiwan}
% \affil[5]{MediaTek, Inc., Hsinchu, Taiwan}

% \renewcommand*{\Affilfont}{\small\it} 
% \renewcommand\Authands{ and }
% \setlength{\affilsep}{0.5em}

\author{
\IEEEauthorblockN{1\textsuperscript{st} Xin~Hong}
\IEEEauthorblockA{\textit{Centre for Quantum Software and Information} \\
\textit{University of Technology Sydney}\\
Sydney, Australia \\
xin.hong@student.uts.edu.au}
\and
\IEEEauthorblockN{2\textsuperscript{nd} Wei-Jia Huang}
\IEEEauthorblockA{\textit{Hon Hai Quantum Computing Research Center} \\
Taipei, Taiwan \\
wei-jia.huang@foxconn.com}
\and
\IEEEauthorblockN{3\textsuperscript{rd} Wei-Chen Chien}
\IEEEauthorblockA{\textit{MediaTek, Inc.} \\
Hsinchu, Taiwan \\
owen.chien@mediatek.com}
\and
\IEEEauthorblockN{4\textsuperscript{th} Yuan Feng}
\IEEEauthorblockA{\textit{Centre for Quantum Software and Information} \\
\textit{University of Technology Sydney}\\
Sydney, Australia \\
yuan.feng@uts.edu.au}
\and
\IEEEauthorblockN{5\textsuperscript{th} Min-Hsiu Hsieh}
\IEEEauthorblockA{\textit{Hon Hai Quantum Computing Research Center} \\
Taipei, Taiwan \\
minhsiuh@gmail.com}
\and
\IEEEauthorblockN{6\textsuperscript{th} Sanjiang Li}
\IEEEauthorblockA{\textit{Centre for Quantum Software and Information} \\
\textit{University of Technology Sydney}\\
Sydney, Australia \\
sanjiang.li@uts.edu.au}
\and
\IEEEauthorblockN{7\textsuperscript{th} Chia-Shun Yeh}
\IEEEauthorblockA{\textit{MediaTek, Inc.} \\
Hsinchu, Taiwan \\
jason.yeh@mediatek.com}
\and
\IEEEauthorblockN{8\textsuperscript{th} Mingsheng~Ying}
\IEEEauthorblockA{\textit{Institute of Software} \\
\textit{Chinese Academy of Sciences}\\
Beijing, China \\
yingms@ios.ac.cn}

}

\maketitle

% \red{Compare: Decision Diagrams for Symbolic Verification of Quantum Circuits}

\begin{abstract}

With the rapid development of quantum computing, automatic verification of quantum circuits becomes more and more important. While several decision diagrams (DDs)  have been introduced in quantum circuit simulation and verification, none of them supports symbolic computation. Algorithmic manipulations of symbolic objects, however, have been identified as crucial, if not indispensable, for several verification tasks. This paper proposes the first decision-diagram approach for operating symbolic objects and verifying quantum circuits with symbolic terms.
As a notable example, our symbolic tensor decision diagrams (symbolic TDD) could verify the functionality of the 160-qubit quantum Fourier transform circuit within three minutes. Moreover, as demonstrated on Bernstein-Vazirani algorithm, Grover's algorithm, and the bit-flip error correction code, the symbolic TDD enables efficient verification of quantum circuits with user-supplied oracles and/or classical controls.

\end{abstract}


\section{Introduction}\label{sec:intro}

% \blue{Ying: The \underline{\textbf{central idea}} of symbolic approach is to leverage the power of (symbolic) logic. 

% The current version of the paper is only simple extension of TDDs to TDDs with symbolic weights. 

% I strongly suggest to add a paragraph to discuss how logical laws can be used in this approach; e.g. x+x'= 1, x.x'=0, ...

% Can we use some (mathmatical or software) tools for  optimising boolean expressions (and thus complex-valued boolean expressins)? Combine them with our implementation?  

% Also, you should seriously consider this problem in case studies. 

% Otherwise, I'm sure the reader will ask how  logical laws can be used.}

% \magenta{The laws such as x+x'= 1, x.x'=0 have been used in the reduction and operations of Boolean functions using BDDs. But the characteristics of Boolean function gives us more opportunities for reducing the symTDDs, for example, the reduction rule shown in Fig. \ref{fig:fur_red}. However, we currently failed to construct a complete set of normalisation and reduction rule based only on these laws.}

% \blue{Ying: I meant systematic use of logical laws, not only these two simple laws.}

In the past several years, quantum hardware has experienced rapid development. In particular, IBM has recently announced their 127-qubit quantum processor ``Eagle'', 433-qubit ``Osprey'' \cite{chow2021ibm} as well as the planned 1121-qubit ``Condor'' in their aggressive roadmap \cite{gambetta2020}. As the scale of quantum devices becomes larger and larger,  the automatic verification of quantum circuits becomes more and more important. 

Decision diagram-based methods are perhaps the most successful verification methods for quantum circuits. Successful examples include, among others, the quantum information decision diagram (QuIDD) \cite{viamontes2003improving}, the quantum multi-value decision diagram (QMDD) \cite{niemann2015qmdds}, and tensor decision diagram (TDD) \cite{hong2020tensor}. A recent work \cite{tsai2021bit} even used binary decision diagrams (BDD) in the verification of quantum circuits. All these decision diagrams can be used to represent quantum states, simulate quantum circuits, and represent the functionality of quantum circuits. However, many quantum circuits still have no compact decision diagram representations.

Take the well-known quantum Fourier transform (QFT) as an example \cite{coppersmith2002approximate,nielsen2002quantum}. QFT can be classically simulated in polynomial time \cite{Aharonov_2008}  and, for any given input state in the computational basis, the size of its decision diagram representation is linear in the number of qubits. However, it requires exponentially more memory resources to fulfil the complete functionality. Indeed, for an $n$-qubit QFT, its QuIDD, TDD, or QMDD representation all have $O(2^n)$ nodes \cite{hong2020tensor,zulehner2018advanced}. This is because they have very different output amplitudes on $2^n$ different computational basis states. For example, if $n=27$, it needs 8 GB of memory to store such a decision diagram.
%$2^{25+1}*2*128/1024/1024/1024/8=$

% \delete{So, is it still possible to represent quantum circuits like QFT compactly with a decision diagram? The answer is affirmative, as} \red{(SL: SymTDD does not provide a representation of the functionality of QFT.)} 
It is worth noting that the decision diagram representations of different input basis states have similar patterns. Instead of simulating on all $2^n$ computational basis states, we propose to interpret an arbitrary input basis state on each qubit as a symbol and simulate the circuit on this symbolised basis state. Fortunately, for QFT, the simulation on this symbolised basis state consumes only polynomial resources. This implies that we may verify the functionality of a quantum circuit with a single (symbolised) simulation. The output state for such a symbolised input state can be elegantly represented by a single TDD-like diagram whose edge weights are symbolic objects instead of complex values. 

%\delete{For the QFT with 160 qubits, such a symbolic decision diagram can be easily constructed within three minutes on a laptop computer. In general, this symTDD provides a novel way for verifying quantum circuits.} %\wcc{Adopting the representation could pave the way for verifying quantum circuits.}

%\delete{On the other hand, many quantum circuits naturally include some symbolic part.} \red{Perhaps more importantly, symbolic objects appear naturally as classical control signals.} \delete{For example, the implementation of some} \red{Many quantum} algorithms, such as Grover's algorithm and Bernstein-Vazirani algorithm, \red{run} according to the classical information they are coping with. \delete{Or in other words, these algorithms are constructed and executed depending on some classical control signals.} \delete{At present, fewer \red{few} tools can be used to analyse such circuits on all possible classical control signals efficiently. However, it} is \delete{sometimes} \red{often} necessary if we want to \delete{check} \red{verify} the correctness of these algorithms. One possible way to do it is also to represent the classical control signal as a symbolic object and then check it on the symbolic level.

Perhaps more importantly, symbolic objects appear naturally as classical control signals in many quantum algorithms, such as Grover's algorithm \cite{grover1996fast} and the Bernstein-Vazirani algorithm (BV) \cite{bernstein1997quantum}. To completely verify the correctness of these algorithms, all possible classical control signals need to be considered. One way of achieving this is to regard classical control signals as symbolic objects and then check the circuit on the symbolic level.

This work formally presents the symbolic tensor decision diagrams (or simply symTDDs) for symbolically executing and representing a quantum circuit, establishes the canonicity, and then demonstrates the efficiency of symTDDs in the simulation and verification of quantum circuits such as QFT, BV, Grover's algorithm, and the bit-flip error correction code circuit. In this paper, to fully utilise the efficiency of TDD, we regard symbolic objects as tensors and represent them as TDDs.

Extending TDDs with symbolic (tensor) weights makes it possible to leverage the power of symbolic logic in a systematic way. On the one hand, Boolean algebra laws such as $x\cdot x=x$, $x\cdot x'=0$, and $x+x'=1$ are used in the normalisation and reduction of symTDDs so that canonicity is guaranteed. On the other hand, the generated symTDDs can be used in extracting distributions of the output states by symbolic computations with these symbolic weights. It is also possible to employ software tools such as SMT solvers like Z3 in the generation of symTDD to check, for example, if some given output state or subspace is reachable.

%searching, for instance, the best input state which maximises the output probability.

\noindent \textit{Related works}. Symbolic verification of quantum circuits has also been explored in \cite{ying2020symbolic}, where a symbolic approach for representing quantum circuits with more general matrix-valued Boolean expressions was proposed. Their approach provides a way to verify quantum circuits with existing techniques and tools developed for the verification of classical logic circuits. In contrast, our work is based on TDDs, which combine the merits of both tensor networks and decision diagrams for representing quantum circuits. In \cite{amy2018towards}, 
the author also discussed the verification of quantum circuits in a symbolic way. Potentially, it could efficiently represent a large class of quantum circuits, while sacrificing uniqueness  (i.e., canonicity). Note that the canonicity of symTDD is guaranteed (cf. Theorem~\ref{thm:canonicity}).

Classical simulation of QFT circuit as an important task has been considered in works such as \cite{Aharonov_2008, Yoran_2007, Browne2007}. These works simulate either on a subset of all possible input states \cite{Browne2007} or in an approximate way \cite{Aharonov_2008, Yoran_2007}. Our scheme, in contrast, can reveal information of the circuit for all computational basis states.

Other verification methods such as equivalence checking and model checking have been considered in \cite{burgholzer2021qcec, burgholzer2020improved, Peham_2022, Ying21mcqsbook}, etc. However, none of them considers executing a quantum circuit with a symbolic input state and verifying a quantum circuit using a decision diagram with symbolic weights.

% Since quantum Fourier transform serves as an essential component for many quantum algorithms, classical simulation of quantum Fourier transform circuit as an important task has been considered in works such as \cite{Aharonov_2008, Yoran_2007, Browne2007}. 

% is more likely an implementation of \cite{amy2018towards} based on decision diagrams, which
% The symTDD provides a compact and canonical way to symbolically execute a quantum circuit and represent quantum circuits with classical control.
% The implementation of symTDD is available online.

% The structure of this paper is as follows. In section~\ref{sec:background}, we give the most basic concepts of quantum circuits and tensor decision diagrams. In section~\ref{sec:STDD}, we discuss complex-valued Boolean functions and present the symTDD. Explicit examples are provided in section~\ref{sec:case} to demonstrate the advantage of symTDD and show how it can be applied in the verification of quantum circuits. Finally, in section~\ref{sec:conclusion}, we conclude our work with a brief outlook on future work.

\section{Background}\label{sec:background}

\noindent \textit{Quantum Computing.}
The most fundamental concept in quantum computing is qubit, the counterpart of the classical bit. While a classical bit can be either 0 or 1, a qubit can be in a superposition of two basis states. Using the Dirac notation, a qubit state can be represented as $\ket{\psi}=a\ket{0}+b\ket{1}$, where $a,b$ are complex numbers with $|a|^2+|b|^2=1$. It is often represented by the vector $\ket{\psi}=[a,b]^T$. In general, an $n$-qubit state can be represented by a $2^n$-dim vector $\ket{\psi}=[\alpha_0,\cdots,\alpha_{2^n-1}]^T$. 

The evolution of quantum states is according to unitary operators, which are also called quantum (logic) gates. The quantum gates used in our paper mainly include the Hadamard gate, the $R_k$ gates, and the Controlled-X gate (aka CNOT or CX gate) defined as follows:
% , see Fig.~\ref{fig:gates} for their matrix representations. 
\begin{equation*}
\resizebox{1\hsize}{!}{%
$
H= \frac{1}{\sqrt{2}}
\big[\begin{smallmatrix}
	1 & 1\\ 
	1 & -1
\end{smallmatrix}
\big]
\quad
R_k=\big[
\begin{smallmatrix} 
	1 & 0\\ 
	0 & e^{2\pi \imath/2^k}
\end{smallmatrix}
\big]
\quad
\mathit{CX}=
\bigg[
\begin{smallmatrix}
	1 & 0 & 0 & 0\\ 
	0 & 1 & 0 & 0\\
	0 & 0 & 0 & 1\\
	0 & 0 & 1 & 0
\end{smallmatrix}
\bigg]
$
}
\end{equation*} 

The reader is referred to, e.g., \cite{nielsen2002quantum} for other standard gates used in this paper. Mathematically, a quantum gate on $n$ qubits can be represented as a $2^n\times 2^n$ unitary matrix.
% \yf{mention that a gate on $n$ qubits has matrix representation with dimension $2^n\times 2^n$}

%For simplicity, we also let the matrix representation of some of the gates used in our paper out of the list as they are standard and can be found in any book on quantum computing.

% \begin{figure}

% % \wcc{Aligned}
% \begin{center}

% \[
% \begin{array}{ r c l }

% \mathit{H\ gate}: &
% \begin{tikzcd}
% &\qw  & \gate{H}  &\qw &\qw
% \end{tikzcd} &
% \frac{1}{\sqrt{2}}
% \begin{bmatrix}
% 	1 & 1\\ 
% 	1 & -1
% \end{bmatrix} \\ [3ex]

% \mathit{X\ gate}: &
% \begin{tikzcd}
% &\qw  & \gate{X}  &\qw &\qw
% \end{tikzcd} &
% \phantom{\frac{1}{\sqrt{2}}}
% \begin{bmatrix} 
% 	0 & 1\\ 
% 	1 & 0
% \end{bmatrix} \\ [3ex]

% \mathit{R_k\ gate}: &
% \begin{tikzcd}
% &\qw  & \gate{R_k}  &\qw &\qw
% \end{tikzcd} &
% \phantom{\frac{1}{\sqrt{2}}}
% \Big[\begin{smallmatrix} 
% 	1 & 0\\ 
% 	0 & e^{2\pi \imath/2^k}
% \end{smallmatrix}
% \Big]\\ [3ex]

% \mathit{CX\ gate}: &
% \begin{tikzcd}
% &\qw  & \ctrl{1}  &\qw &\qw \\
% &\qw  & \gate{X}  &\qw &\qw
% \end{tikzcd} &
% \phantom{\frac{1}{\sqrt{2}}}
% \bigg[
% \begin{smallmatrix}
% 	1 & 0 & 0 & 0\\ 
% 	0 & 1 & 0 & 0\\
% 	0 & 0 & 0 & 1\\
% 	0 & 0 & 1 & 0
% \end{smallmatrix}
% \bigg]

% \end{array}
% \]

% \end{center}
% \caption{The matrix representations of the Hadamard gate, the Pauli X gate, the $R_k$ gates and the Controlled-X gate. \yf{This figure can be deleted: put the matrix representations in the main text}}
% \label{fig:gates}
% \end{figure}

% Quantum logic gates can be classified single-qubit gates and two-qubit gates.
A quantum circuit is a sequence of quantum gates. For example, Fig. \ref{fig:exp-for-quantum-circuit} depicts a quantum circuit of the QFT on three qubits. Given a computational basis input $\ket{s_0}\ket{s_1}\ket{s_2}$, the expected output is $\frac{1}{2\sqrt{2}}(\ket{0}+e^{2\pi \imath 0.s_0s_1s_2}\ket{1})(\ket{0}+e^{2\pi \imath 0.s_1s_2}\ket{1})(\ket{0}+e^{2\pi \imath 0.s_2}\ket{1})$. Note that we use $\imath$ to represent the imaginary unit. 

% Using quantum circuits can implement many useful quantum algorithms that can achieve quantum speed up. For example, quantum Fourier transform\cite{coppersmith2002approximate}, Grover's algorithm\cite{grover1996fast}, Shor's algorithm\cite{shor1994algorithms} and so on. However, quantum computing is still in the Noisy Intermediate-Scale Quantum (NISQ) \cite{preskill2018quantum} stage. These algorithms are still not able to demonstrate the power. So we still need quantum simulators to simulate these algorithms.

% \begin{figure}
% \scalebox{0.8}{
% \centerline{
% \Qcircuit @C=1em @R=0.9em {
% \lstick{\ket{x_0}} & \gate{H} &\ctrl{1}   &\ctrl{2}   &\qw      &\qw         &\qw      &\qw & \rstick{\ket{0}+e^{2\pi i 0.x_0x_1x_2}\ket{1}}\\
% \lstick{\ket{x_1}} & \qw      &\gate{R_2} &\qw        &\gate{H} &\ctrl{1}    &\qw      &\qw & \rstick{\ket{0}+e^{2\pi i 0.x_1x_2}\ket{1}}\\
% \lstick{\ket{x_2}} & \qw      &\qw        &\gate{R_3} &\qw      &\gate{R_2}  &\gate{H} &\qw & \rstick{\ket{0}+e^{2\pi i 0.x_2}\ket{1}}\\
% }
% }
% }
% \caption{The circuit for 3-qubits quantum Fourier transform.}
% \label{exp-for-quantum-circuit}
% \end{figure}

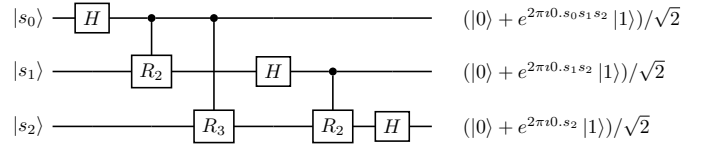
\begin{figure}
\centerline{
\scalebox{0.75}{
\begin{quantikz}[column sep=0.4cm,row sep=0.4cm]
\lstick{$\ket{s_0}$} & \gate{H} &\ctrl{1} &\ctrl{2}   &\qw      &\qw       &\qw      &\qw & \rstick{$(\ket{0}+e^{2\pi \imath 0.s_0s_1s_2}\ket{1})/\sqrt{2}$}\\
\lstick{$\ket{s_1}$} & \qw      &\gate{R_2} &\qw        &\gate{H} &\ctrl{1}  &\qw      &\qw & \rstick{$(\ket{0}+e^{2\pi \imath 0.s_1s_2}\ket{1})/\sqrt{2}$}\\
\lstick{$\ket{s_2}$} & \qw      &\qw      &\gate{R_3}   &\qw      &\gate{R_2}  &\gate{H} &\qw & \rstick{$(\ket{0}+e^{2\pi \imath 0.s_2}\ket{1})/\sqrt{2}$}\\
\end{quantikz}
}
}
\caption{The circuit for 3-qubit quantum Fourier transform.}
\label{fig:exp-for-quantum-circuit}
\end{figure}

% \delete{Decision diagrams have been widely used in the verification of quantum circuits\cite{burgholzer2020improved,yamashita2010fast}. Among them, the most commonly used decision diagrams include the quantum information decision diagram (QuIDD) \cite{viamontes2003improving}, the quantum multi-value decision diagram (QMDD)\cite{niemann2015qmdds} and tensor decision diagram (TDD) \cite{hong2020tensor}. }

%Since there is an intrinsic similarity among the decision diagrams mentioned in section \ref{sec:intro}, we select TDD as our starting point. 
\vspace*{2mm}
\noindent\textit{Tensor Decision Diagram (TDD).}
Tensor network methods play a fundamental role in machine learning and quantum physics. By combining the merits of both tensor networks and decision diagram, TDDs \cite{hong2020tensor} provide a compact and canonical way to represent quantum circuits, which are a special class of tensor networks. TDDs have also been applied in equivalence checking of quantum circuits \cite{hong2021approximate}.

% WCC: use ${}$ to make unbreakable expression
% \yf{what do you mean by `linear' map here?}
A rank $n$ complex-valued tensor is a map ${\phi(q_1\ldots q_n): \{0,1\}^n \to \mathbb{C}}$, where $I=\{q_1,\cdots,q_n\}$ is the set of indices of the tensor and $\phi(a_1, \ldots, a_n)$ its value for the evaluation $q_1=a_1,\cdots,q_n=a_n$. In the following, we also denote it as $\phi (\bold{q})$ for simplicity. The classical Boole-Shannon expansion \cite{Shannon_1938} also extends to tensors:
${\phi = q_k^{\prime} \cdot \phi|_{q_k = 0} + q_k \cdot \phi|_{q_k = 1}}$, where $\phi|_{q_k = a_k}$ is a rank $n-1$ tensor with indices $\{q_1,\cdots,q_{k-1},q_{k+1},\cdots,q_n\}$ and $\phi|_{q_k = a_k}(a_1,\cdots, a_{k-1},a_{k+1},\cdots,a_{n})=\phi(a_1,\cdots,a_n)$, and $q_k^{\prime}$ represents the complement of $q_k$. Obviously, an $n$-qubit quantum state can be seen as a rank $n$ tensor, and an $n$-qubit quantum gate can be seen as a rank $2n$ tensor. When associating to quantum states, $q_k$ normally represents the index of the $k$-th qubit and we use $q_k^i$ and $q_k^o$ to represent the indices of the input and output of a quantum gate corresponding to the $k$-th qubit.

% \yf{why use $b$ here? You are using $a_i$ to denote the value of $q_i$. So it's better to use $a$ instead of $b$ here}
% the restriction of the tensor $\phi$ to the case when $q_i=a_i$, $a_i \in \{0,1\}$ \yf{explain what $q_i'$ is.} \yf{You also should explain what `restriction' of a tensor is. Say,  $\phi|_{q_i = b}$ is a tensor with index set $I-\{q_i\}$ and its value..., etc.}.
%   \yf{ask yourself what `which share a common index set $\bold{r}$' means}

There are mainly two tensor operations: addition and contraction. Addition is defined for two tensors $\phi,\psi$ with the same index set, and $(\phi+\psi)(\bold{a})=\phi(\bold{a})+\psi(\bold{a})$ for any assignments $\bold{a}$ of the indices. Tensor contraction can be defined for two tensors with different index sets. Let $\gamma(\bold{q}, \bold{r})$ and $\xi(\bold{r}, \bold{s})$ be two tensors of which the indices share a common part $\bold{r}$. Their contraction is the sum of product over every assignment $\bold{c}$ of indices in $\bold{r}$. The new tensor $\phi(\bold{q}, \bold{s})$ with assignments $\bold{a}$ and $\bold{b}$ hence becomes
\begin{equation}\label{eq:contdef}
	\phi(\bold{a}, \bold{b})=\sum_{\bold{c}\in \{0,1\}^{\bold{r}}}{\gamma(\bold{a}, \bold{c})\cdot \xi(\bold{c}, \bold{b})}.
\end{equation}

In this paper, we will also use the Hadamard product $\odot$ of two tensors, which is defined for two tensors with the same indices and $(\phi\odot\psi)(\bold{a})=\phi(\bold{a})\cdot \psi(\bold{a})$. All complex-valued tensors with the same index set form a commutative ring with the two operations $+$ and $\odot$.
 
%  \yf{connect the general notion of `tensor' with the specific `state' and `gate' introduced in the last subsection.}

%Then, the tensor decision diagram represent it as a dichotomous way ${\phi = \overline{x_i} \cdot \phi|_{x_i = 0} + x \cdot \phi|_{x_i = 1}}$, where $\phi|_{x_i = b}$ is the restriction of the tensor $\phi$ to the case when $x_i=b$, $b\in \{0,1\}$. \lsj{This sentence does not explain what is a TDD.}

A TDD over an index set $I$ is a directed acyclic graph $\mathcal{T}=(V,E,idx,val,w)$, where the node set $V$ consists of non-terminal nodes in $V_N$ and terminal ones in $V_T$ and each non-terminal node $v$ has two child nodes $low(v)$ and $high(v)$ (also called the 0- and 1-successors of $v$); the associated functions are $idx:V_N\rightarrow I$, $val:V_T\rightarrow\mathbb{C}$, and $w:E\rightarrow\mathbb{C}$ which are respectively assigned node index, terminal node value, and edge weight \cite{hong2020tensor}. For any internal node $v$, we also call the edge from $v$ to $low(v)$ (resp. $high(v))$ the low-edge or 0-edge (resp. the high-edge or 1-edge). 
A general TDD may be transformed into a reduced and normalised one, which is canonical w.r.t. a fixed index order, by a series of normalisation and reduction procedures. Usually, all terminal nodes in a TDD are merged to a unique one labelled with 1. Then the tensor value of an evaluation $q_1=a_1,\cdots,q_n=a_n$  can be obtained by multiplying the weights of the edges along the path assigned by the evaluation.
% \yf{how does this description apply to reduced TDDs? Say, the left-most path of the TDD in Fig. \ref{fig:TDD_exp}?}.

% \yf{specifically mentioning this property gives the impression that reduced and normalised TDD means that there is a unique terminal node associated with the value $1$} 

\begin{figure}
    \centering
    \includegraphics[width=0.08\textwidth]{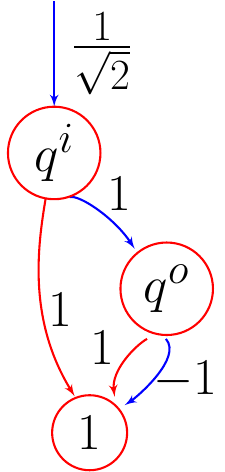}
    \caption{The canonical TDD of the $H$ gate w.r.t. the index order $q^i\prec q^o$, where $q^i$ and $q^o$ are two internal nodes and we have a unique terminal node.}
    \label{fig:TDD_exp}
\end{figure}

% the \yf{not `the', but `a'. Or `the canonical'}
\begin{exam}
Fig. \ref{fig:TDD_exp} gives the canonical TDD of the $H$ gate w.r.t. the index order $q^i\prec q^o$. Each matrix element of $H$ corresponds to the product of the weights along a maximal path from the root to the terminal. For any internal node $u$, the outgoing red (blue, resp.) edge connects its 0-successor (1-successor, resp.), denoting that the index takes value $0$ ($1$, resp.). The red edge of node $q^i$ implies that all elements of the first column of this matrix are $\frac{1}{\sqrt{2}}\cdot 1=\frac{1}{\sqrt{2}}$. Similarly, the blue edge leading to node $q^o$ implies that the second column of the matrix is $\frac{1}{\sqrt{2}}\cdot 1\cdot [1,-1]^T=[\frac{1}{\sqrt{2}},-\frac{1}{\sqrt{2}}]^T$.
% \yf{here you use dotted vs solid lines to distinguish 0 and 1. In STDD different colors are used. Better to be consistent.}
\end{exam}
%\yf{you may refer to your previous TDD paper or QMDD papers to see how to make the explanation of TDDs more clear.}

% \yf{distinct indices?} from the starting edge (weight $=1/\sqrt{2}$) to the terminal node (value $=1$) \yf{the whole sentence is vague}

\section{Symbolic TDD}\label{sec:STDD}

%In this section, we introduce the symTDD. 
Symbolic TDD (symTDD) is TDD with all weights being symbolic objects instead of complex values. In this paper, these symbolic objects are complex-valued tensors whose indices are symbols from either the input symbolised basis state or classical control signals.

\subsection{Motivating Examples}\label{sec:motivation}

Similar to the 3-qubit QFT case shown in Fig.~\ref{fig:exp-for-quantum-circuit}, the output states (functionality)  of a quantum circuit $C$ on all $2^n$ computational basis states can be (symbolically) computed by executing $C$ on a single symbolised basis state. Furthermore, symTDD can also be used to represent the oracle of many quantum algorithms, including Grover's algorithm and Bernstein-Vazirani algorithm.

% 1. Use symTDD to simulate the execution of a quantum circuit with an unknown computational basis input.

% 2. Use symTDD to describe the (unknown) oracle of Grover algorithm.

Fig.~\ref{fig:exp_stdd} illustrates several symTDDs, which represent (a) the symbolic input basis state $\ket{s}=s^{\prime}\ket{0}+s\ket{1}$, (b) $H\ket{s}$, the resultant state of applying a Hadamard gate on $\ket{s}$, (c)  $R_2\ket{s}$,  the resultant state of applying a rotation gate $R_2$ on $\ket{s}$, and (d) the resultant state of applying a 2-qubit Grover oracle on $\ket{001}$, where $\ket{1}$ is the state of the ancilla qubit. Note that $H\ket{s}=\frac{1}{\sqrt{2}}(\ket{0}+(s^{\prime}-s)\ket{1})$. If we take $s=0$, then Fig.~\ref{fig:exp_stdd}(a,b) exactly give the TDD representations of the states $\ket{0}$ and $\ket{+}=H\ket{0}$. A similar observation applies to the case $s=1$. This implies that the information of the circuit for all input computational basis states is captured by a single simulation on the symbolised basis state. Note that all weights here are tensors: $s$ and $s^{\prime}$ (the complement of $s$) represent the tensors with only one index $s$, and $\phi(s)=s$ and $\xi(s)=1-s$; $s_0^{\prime}s_1^{\prime}$ and $s_0s_1^{\prime}+s_1$ represent the rank 2 tensors with indices $s_0, s_1$, where $\phi(s_0,s_1)=(1-s_0)\cdot(1-s_1)$ and $\xi(s_0,s_1)=s_0\cdot (1-s_1)+s_1$.
% \yf{You should mention that the weights here are all tensors. Better to explain why, say, $s_0s_1' + s_1$ is a tensor.}

% From Fig.~\ref{fig:exp_stdd}(d), we observe that, when the input states of the first two qubits are both $\ket{0}$ (i.e., following the rightmost path from $q_0$ to $q_2$ via $q_1$), the ancilla qubit will be flipped (i.e., $s_0's_1'=1$ and $s_0s'_1+s_1=0$) if and only if $s_0=0$ and $s_1=0$. \yf{what's the point of this paragraph?}

% In addition, Fig. \ref{fig:exp_stdd} (c) and (d) give the symTDD representation of applying a $R_2$ and $R_3$ gate to the input state $\ket{s}$.

\begin{figure}
    \centering
    \subfigure[]{
    \includegraphics[height=0.15\textwidth]{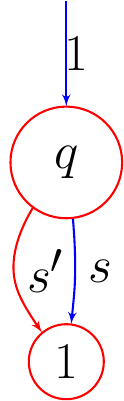}
    }
    \subfigure[]{
    \includegraphics[height=0.15\textwidth]{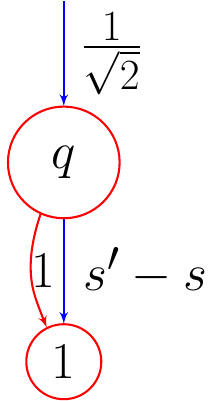}
    }
    \subfigure[]{
    \includegraphics[height=0.15\textwidth]{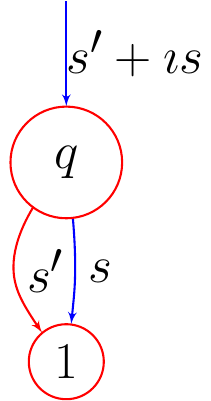}
    }
    % \subfigure[]{
    % \includegraphics[height=0.16\textwidth]{figures/exp3.pdf}
    % }    
    \subfigure[]{
    \includegraphics[height=0.16\textwidth]{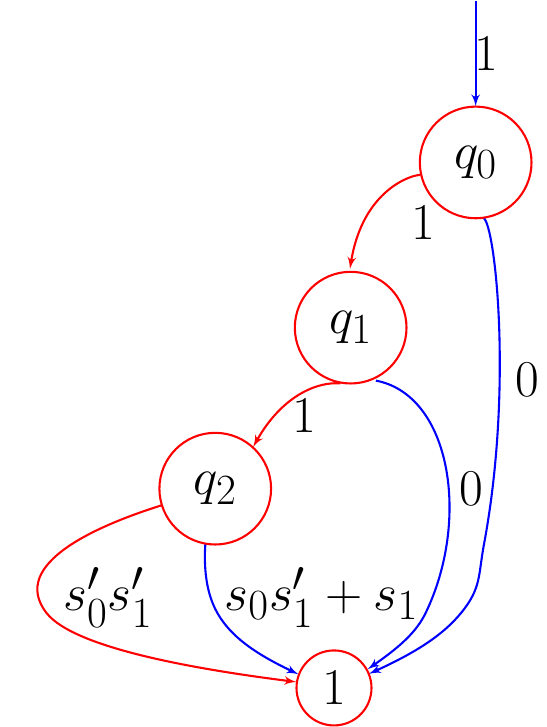}
    }    
    \caption{symTDD representations of (a)  a 1-qubit symbolised computational basis state $\ket{s}$; (b)  $H\ket{s}$; (c) $R_2\ket{s}$; (d) the output state of the 2-qubit Grover oracle applied on $\ket{001}$, where $\ket{1}$ is the state of the ancilla qubit.}
    %  \red{SL: Add examples like $R_k\ket{s_0}$ for $k=2,3$ here.}
    \label{fig:exp_stdd}
\end{figure}

\subsection{Symbolic TDD and Operations}

SymTDD is an extension of TDD \cite{hong2020tensor} in that its weights are, instead of complex values, complex-valued tensors over a set of indices (i.e., Boolean symbols) $S=\{s_0, \cdots, s_{m-1}\}$. Since a TDD represents a tensor, a symTDD represents a tensor-valued tensor $\phi$ (a map from $\{0,1\}^n$ to tensors), or more precisely, a map in $\{0,1\}^n \to \{0,1\}^m \to \mathbb{C}$. For clarity, let $I=\{q_0,\ldots,q_{n-1}\}$ be the set of indices other than those in $S$. For convenience, we often address indices in $I$ and $S$ as, respectively, quantum and classical indices. In addition, we call $\phi$ an $(I,S)$-tensor with rank $n$ or simply an $(m,n)$-tensor. This representation gives us more flexibility for representing and analysing a quantum circuit. For example, let $\bold{s}$ be the indices that appear in both an $n$-qubit input state $\ket{\psi}$ and an $n$-qubit quantum circuit $U$ (note that a tensor with index set $S$ can always be regarded as a tensor with index set $S'$ with $S'\supseteq S$). That is, $\ket{\psi}=[f_0(\bold{s}),\cdots,f_{2^n-1}(\bold{s})]^T$ and $U = [u_{ij}(\bold{s}): i,j=0,\ldots, 2^n-1]$
where both $f_j$ and $u_{ij}$ are tensors in $\{0,1\}^m \to \mathbb{C}$. Note that 
$U\ket{\psi} = [g_0(\bold{s}),\cdots,g_{2^n-1}(\bold{s})]^T$ where for each $i$, 
$$g_i(\bold{s}) = \sum\nolimits_{j=0}^{2^n-1} u_{ij}(\bold{s})\odot f_j(\bold{s}).$$
Then $g_i$ is again a tensor in $\{0,1\}^m \to \mathbb{C}$, and finally, $U\ket{\psi}$ is a symbolic tensor in $\{0,1\}^n \to \{0,1\}^m \to \mathbb{C}$.
%\delete{Since a rank $n$ tensor can be represented by a map in $\{0,1\}^{(n-m)} \to \{0,1\}^m \to \mathbb{C}$ for any $m \in \{0, \cdots n-1\}$ \yf{this is not a complete sentence}. When $m=0$, its symTDD representation is exactly a TDD.}

% Following, we take $S$ to be the collection of all complex-valued tensor over a set of indices $S=\{s_0, \cdots, s_{m-1}\}$. Then, $S$ forms a commutative ring. Let, $f,\ g$ be two elements in $S$, then $f+g$ and $f\cdot g$ are both an element in $S$, where $(f+g)(\bold{a})=f(\bold{a})+g(\bold{a})$ and $(f\cdot g)(\bold{a}) = f(\bold{a})\cdot g(\bold{a})$, for any assignment $a$ of variables in $X$. 

% Using the Boole-Shannon expansion, $f=x_i^{\prime}\cdot f|_{x_i=0}+x_i\cdot f|_{x_i=1}$ and $g=x_i^{\prime}\cdot g|_{x_i=0}+x_i\cdot g|_{x_i=1}$. Then, $f+g=x_i^{\prime}\cdot (f|_{x_i=0}+g|_{x_i=0})+x_i\cdot (f|_{x_i=1}+g|_{x_i=1})$ and $f\cdot g=x_i^{\prime}\cdot (f|_{x_i=0}\cdot g|_{x_i=0})+x_i\cdot (f|_{x_i=1}\cdot g|_{x_i=1})$, which indicates a recursive implementation of both the addition and multiplication of two tensors. 

% \subsection{Operations}
The overall representation and calculation are similar to that in TDD, except that the addition and multiplication of complex numbers used in TDD (Alg.s~1$\sim$3 in \cite{hong2020tensor}) should be replaced by the addition and multiplication of complex-valued tensors. 
%\delete{They have the same complexity, except that every operation should be augmented by the complexity of manipulating two weighted BDDs.}

% The complexities of the addition and contraction of two TDDs with complex value weight are $\mathcal{O}(|\mathcal{F}|\cdot|\mathcal{G}|)$ and $\mathcal{O}(|\mathcal{F}|^2\cdot|\mathcal{G}|^2)$ \cite{hong2020tensor}, respectively, while $|\mathcal{F}|$ and $|\mathcal{G}|$ are the numbers of nodes in the two TDDs. The complexities of the addition and multiplication of two weighted BDDs are both $\mathcal{O}(|\mathcal{B}_1|\cdot|\mathcal{B}_2|)$ \cite{bryant1986graph}, while $|\mathcal{B}_1|$ and $|\mathcal{B}_2|$ are the number of nodes in the two BDDs. Thus, The complexities of the addition and contraction of two symTDDs are $\mathcal{O}(|\mathcal{F}|\cdot|\mathcal{G}| \cdot |\mathcal{B}_1|\cdot|\mathcal{B}_2|)$ and $\mathcal{O}(|\mathcal{F}|^2\cdot|\mathcal{G}|^2 \cdot |\mathcal{B}_1|\cdot|\mathcal{B}_2|)$, while $|\mathcal{F}|$ and $|\mathcal{G}|$ are the numbers of nodes in the two symTDDs, and $|\mathcal{B}_1|$ and $|\mathcal{B}_2|$ are the number of nodes of the biggest BDDs appeared on the weights of the two symTDDs.

\subsection{Normalisation}
Normalisation is required to make symTDD representations canonical. The main idea is to extract a (greatest) common part from the weights on the two outgoing edges of every internal node and multiply it to their incoming edges, such that no more common part can be further extracted from the two remaining weights.

Assume that all tensors are over the same index set $S$. Let $f,g$ be the weights on the outgoing edges of internal node $v$. We define the extracted tensor $h$ and the remaining parts $f^*$ and $g^*$ in a term-by-term manner. For any assignment $\bold{a}$ of $S$, we set $h(\bold{a})$ as $f(\bold{a})$ if $f(\bold{a})\neq 0$, and $g(\bold{a})$ otherwise. Furthermore,
\begin{equation*}\label{eq:loc_norm}
\resizebox{\hsize}{!}{%
$
f^*(\bold{a})=
\begin{cases}
1 & \text{if}\ f(\bold{a})\neq 0\\
0 & \text{otherwise},
\end{cases}
\quad
g^*(\bold{a})=
\begin{cases}
g(\bold{a})/f(\bold{a}) & \text{if}\ f(\bold{a})\neq 0\\
1 & \text{if}\ f(\bold{a})= 0,\ g(\bold{a})\neq 0\\
0   & \text{if}\ f(\bold{a})=g(\bold{a})= 0
\end{cases}
$
}
\end{equation*} 
% \yf{So this `local normalisation' step will take exponential time? Given that you use DD to represent $f$ and $g$, it's better to construct local normalisation directly based on DDs, not the vector representations. Consequently, the complexity will depend on the sizes of the DDs, not necessarily being exponential.}
We call this the local normalisation procedure and denote the result as $loc\_norm(f,g)=(h,f^*,g^*)$. For example, let $f=[2\imath,0,1+\imath,0]$ and $g=[0,1,1-\imath,0]$ be two tensors over $S$ in vector representations. They will be normalised to $(h,f^*,g^*)$ with $h=[2\imath,1,1+\imath,0]$, $f^*=[1,0,1,0]$, and $g^*=[0,1,\frac{1-\imath}{1+\imath},0]$. 

%This local normalisation procedure can also be conducted directly using TDD. Actually, the whole procedure is similar to TDD addition and multiplication as they all execute in a term-by-term way and are implemented recursively.

Let $\supp{f}=\{\bold{a} \in \{0,1\}^n : f(\bold{a})\neq 0\}$. For a local normalisation $loc\_norm(f,g)=(h,f^*,g^*)$, we have  $\supp{h}=\supp{f} \cup \supp{g}$, $\supp{f^*}=\supp{f}$, and $\supp{g^*}=\supp{g}$. If $\supp{f}=\{0,1\}^n$, then $h=f$, $f^{*}=1$ and $g^*=\big[\frac{g_0}{f_0},\cdots,\frac{g_{2^n-1}}{f_{2^n-1}}\big]$. We also have the following useful lemma.
\begin{lem}\label{lemm:loc}
If $loc\_norm(f,g)=(h,f^*,g^*)$ and $\supp{h} \subseteq supp(h^*)$, then $loc\_norm(h^*\odot f, h^* \odot g)=(h^*\odot h,f^*,g^*)$; if $loc\_norm(h\odot f,h\odot g)=(h,f,g)$ and $\supp{h} = supp(h^*)$, then $loc\_norm(h^*\odot f, h^* \odot g)=(h^*,f,g)$.
\end{lem}

By using this local normalisation scheme, we can \emph{fully} normalise a symTDD $\mathcal{T}$. A symTDD $\mathcal{T}$ is called \textbf{fully normalised} if for each internal node $v$ with two outgoing edges weights $f_v$ and $g_v$ and any path leading to it with accumulated weights $h_v$ we have $loc\_norm(h_v\odot f_v,h_v\odot g_v)=(h_v,f_v,g_v)$. The fully normalised form reveals the invariance under local normalisation. Typically, a symTDD can be fully normalised by pushing all weights down to the bottom and then performing the local normalisation procedure on each internal node from bottom to top. During this process, however, we may need to split a lot of nodes causing additional time and memory overhead. But the following theorem shows that the normalisation procedure can be done more efficiently for the cases considered in this paper.

% \begin{defn}
% \end{defn}

\begin{thm}[normalisation]\label{thm:normalisation}
Let $\mathcal{T}$ be a symTDD with quantum and classical index sets $I$ and $S$. For each internal node $v$ of $\mathcal{T}$ and any incoming edge $e$ of $v$, write $h_v,f_v,g_v$ for, respectively, the weights of $e$ and the two outgoing edges of $v$.
Then $\mathcal{T}$ is fully normalised if and only if $loc\_norm(h_v\odot f_v,h_v\odot g_v)=(h_v,f_v,g_v)$ for each internal node $v$ and any incoming edge $e$ of $v$.
\end{thm}
\begin{proof} 
If $\mathcal{T}$ has no internal nodes, this result is clearly true.

In the following, let $\mathcal{T}$ be a symTDD and $v$ an internal node of $\mathcal{T}$. For any  incoming edge  $e$ of $v$, suppose $h_v$ is weight on $e=(u,v)$ and $h_v^*$ the corresponding weight accumulated by multiplying all the weights leading to $v$ through $e$. In addition, write $f_v$ and $g_v$ for the weights on the two outgoing edges of $v$.

If $v$ is the root node of $\mathcal{T}$, we have $h_v=h_v^*$. Thus, $loc\_norm(h_v\odot f_v,h_v\odot g_v)=(h_v,f_v,g_v)$ iff $loc\_norm(h^*_v\odot f_v,h^*_v\odot g_v)=(h^*_v,f_v,g_v)$. Otherwise, suppose $h_v^*=h_0\odot \cdots \odot h_k \odot h_v$. Note that here $u$ is the parent node of $v$ s.t. $e=(u,v)$. Then $h_u^*=h_0\odot \cdots \odot h_k$.

Then, suppose $\mathcal{T}$ is fully normalised. We have $loc\_norm(h_v^*\odot f_v,h_v^*\odot g_v)=(h_v^*,f_v,g_v)$. Then $supp(h_v^*)=supp(f_v)\cup supp(g_v)$. By applying this property for the node $u$, we have $supp(h_u^*) \supseteq supp(h_v)$. As a result, we must have $supp(h_v)=supp(h_v^*)$. According to Lemma~\ref{lemm:loc}, we have $loc\_norm(h_v\odot f_v,h_v\odot g_v)=(h_v,f_v,g_v)$.

For the sufficiency part, we have $loc\_norm(h_v\odot f_v,h_v\odot g_v)=(h_v,f_v,g_v)$. According to the local normalisation scheme,  $supp(f_v) = supp(h_v\odot f_v) \subseteq supp(h_v)$. Also, we have $supp(h_v)\subseteq supp(h_k)\subseteq \cdots \subseteq supp(h_0)$. Thus, $supp(h_v^*)=supp(h_0\odot \cdots \odot h_k \odot h_v)=supp(h_v)$. Then according to lemma \ref{lemm:loc}, we have $loc\_norm(h_v^*\odot f_v,h_v^*\odot g_v)=(h_v^*,f_v,g_v)$.

\end{proof}

Now, we have a more efficient procedure for normalising a symTDD $\mathcal{T}$. That is, we first do a local normalisation for every internal node of $\mathcal{T}$ from top to bottom and then do another local normalisation for every internal node of $\mathcal{T}$ from bottom to top. The top-down procedure ensures that, for every internal node $v$, the support of the weight on any incoming edge of $v$ contains the supports of the weights of $v$'s outgoing edges; and the bottom-up procedure ensures that $loc\_norm(h_v\odot f_v, h_v\odot g_v)= (h_v, f_v, g_v)$ for every internal node $v$. As we don't need to push every weight down to the bottom, the probability of node splitting and hence the space and time consumption is significantly reduced.

\subsection{Reduction}\label{sec:red}

%Another thing that should be considered is the reduction. 
Reduction is a process to merge nodes that represent the same tensors so that the decision diagram can be as compact as possible. The following two reduction rules are commonly used for various kinds of decision diagrams, including BDDs and TDDs. 

\begin{itemize}
    \item RR1: Delete a node $v$ if its 0- and 1-successors are both $w$ and its low- and high-edges have the same weight $f$. Meanwhile, redirect the incoming edge of $v$ to $w$.
    \item RR2: Merge two nodes if they have the same index, the same 0- and 1-successors, and the same weights on the corresponding edges. 
\end{itemize}
Note in RR1 we have $h\odot f=h$ if the local normalisation has been conducted on $v$, where $h$ is the weight on the incoming edge of $v$.
%The reduction rules are based on the Boolean laws $x+x^{\prime}=1$ and $x\cdot x^{\prime}=0$. Thus, the property of a node in a \cyan{reduced normalised symTDD} is completely characterised by its two successors as well as the weights on the corresponding edges.

% \red{SL: If RR3 is only used in the examples, not necessary to include it here. May put it in the appendix if it does not affect the proof of the canonicity.} \magenta{RR3 is necessary for making the representation of Toffoli circuits tower form. I think it is powerful, since it can include a lot of cases. But I still do not know how powerful it is. In addition, RR2 can be seen as a special case of the generalised RR3.}

% \red{SL: I still think RR3 is a distraction to the reader. Will the reviewer care about this more powerful reduction method? I am sure they don't need to know this. If they are interested, they can check the appendix.}

%In addition to RR2, we also give a further reduction rule in Appendix.This reduction is an extension of RR2 but more powerful. We will use it in the calculation process to make the representation more compact. But since it will change the weights and affect the normalisation, we will not use it in making a canonical representation. 

\begin{figure}
    \centering
    \subfigure[]{
    \includegraphics[height=0.16\textwidth]{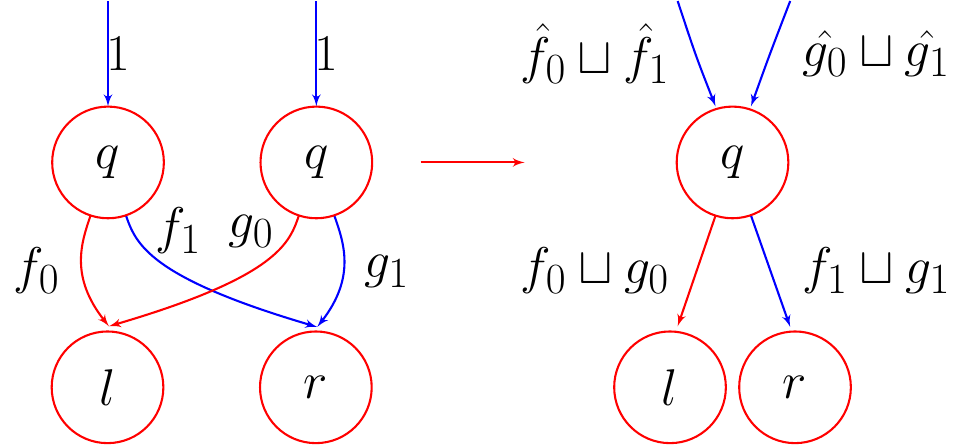}
    }
    \subfigure[]{
    \includegraphics[height=0.17\textwidth]{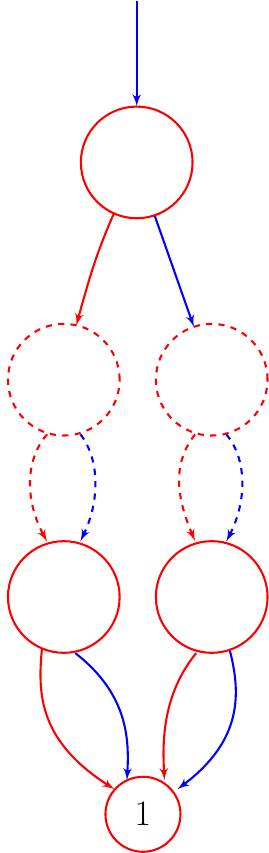}
    }
    \caption{Reduction rule RR3.}
    \label{fig:fur_red}
\end{figure}

While the above two reduction rules are sufficient, the  reduction rule RR3 as specified in  Fig.~\ref{fig:fur_red} (a) is sometimes more powerful. To apply RR2, we require the weights on the corresponding low- and high edges to be identical, which are often too restricted for tensors. From the left side of Fig.~\ref{fig:fur_red} (a), we can see that RR3 is applicable when two nodes (i) have the same index and the same 0- and 1-successors; (ii) $g_i(\bold{a})=f_i(\bold{a})$ for all $\bold{a} \in \supp{f_i} \cap \supp{g_i}$, where $f_i$ and $g_i$ are the weights on the outgoing edge to the $i$-successor for $i \in \{0,1\}$; and (iii) $\supp{f_{0}} \cap \supp{g_{1}} = \supp{f_{1}} \cap \supp{g_{0}}$. In the right side of Fig.~\ref{fig:fur_red} (a), 
$\widehat{f}$ denotes the weight tensor which takes value 1 at each element in  $\supp{f}$ and 0 otherwise; $f \sqcup g$ denotes the tensor which takes the value as $f(\bold{a})$ at each $\bold{a}\in\supp{f}$ and $g(\bold{a})$ otherwise. Here, we assume that $f$ and $g$ have the same value on their common support when using this operation. The basic idea here is that $(\widehat{f_0}\sqcup \widehat{f_1}) \odot (f_0 \sqcup g_0)=f_0$, $(\widehat{f_0}\sqcup \widehat{f_1}) \odot (f_1 \sqcup g_1)=f_1$, $(\widehat{g_0}\sqcup \widehat{g_1}) \odot (f_0 \sqcup g_0)=g_0$ and $(\widehat{g_0}\sqcup \widehat{g_1}) \odot (f_1 \sqcup g_1)=g_1$. In RR3, the weights $\widehat{f_0}\sqcup \widehat{f_1}$ and $\widehat{g_0}\sqcup \widehat{g_1}$ are used as filters indicating when the corresponding edges will lead to a node which represents a tensor with non-zero values. Then $f_0 \sqcup g_0$ and $f_1 \sqcup g_1$ specify the corresponding values. RR3 is more powerful than RR2 and is essential for making quantum circuits, such as the Toffoli circuit, in tower forms. Since RR3 may change the weights on the edges and destroy the normalisation of a symTDD, we only use it at specific times as indicated in Fig.~\ref{fig:fur_red} (b). That is, we will only merge two nodes using RR3 when all the paths leading to these two nodes come from the same node. The nodes with dotted lines mean that there can be zero or many such nodes.

\subsection{Canonicity}
% \delete{We note that this representation can be made canonical if we use RR3 carefully.} \lsj{If not used in the main text, then delete it.}
By using the normalisation and reduction rules, we are able to construct a canonical symTDD representation for every tensor-valued tensor. In the following, we only consider the use of RR1 and RR2 for simplicity. A fully normalised symTDD is reduced when no reduction rule (RR1 or RR2) can be further applied. The following theorem shows that there is, up to isomorphism and a given index order, a unique reduced symTDD for each tensor. Two symTDDs $\mathcal{F}$ and $\mathcal{G}$ w.r.t. the same index order are \emph{isomorphic} if there is a bijection $\sigma$ between the node sets  of $\mathcal{F}$ and $\mathcal{G}$ such that, for each node $v$, $v'=\sigma(v)$ and $v$ have the same index, the same weights on their incoming and outgoing edges, and $\sigma(low(v))=low(v')$, $\sigma(high(v))=high(v')$.

\begin{thm}[canonicity]\label{thm:canonicity}
Let $S=\{s_0,\ldots,s_{m-1}\}$ be a set of classical indices and $I=\{q_0,\ldots,q_{n-1}\}$ a set of quantum indices. Suppose $\phi$ is an $(I,S)$-tensor.  Given any index order $\prec$ of $I$, $\phi$ has a unique reduced symTDD representation w.r.t. $\prec$ up to isomorphism.
\end{thm}
\begin{proof} 
We prove this by using induction on the rank of the tensor. Let $\phi = s\in [2^m \to \mathbb{C}]$ be a rank 0 tensor. Any reduced symTDD of $\phi$ has only one node labelled with $1$ and has weight $s$. Suppose the statement holds for any rank $k$ tensor. Let $\phi$ be a rank $k+1$ tensor and $\prec$ any index order on $I$. Suppose $\mathcal{F}$, $\mathcal{G}$ are two reduced symTDDs that represent $\phi$ w.r.t. $\prec$. We show they are isomorphic. Let $q$ be the first index under $\prec$. Without loss of generality, we assume $\phi|_{q=0}\neq \phi|_{q=1}$. That is, $q$ is essential in $\phi$.

Write $r_\mathcal{F}$ and $r_\mathcal{G}$ for the root nodes of $\mathcal{F}$ and $\mathcal{G}$.  Then $q$ is the index of both $r_\mathcal{F}$ and $r_\mathcal{G}$. Let $h_1,f_1,g_1$ and $h_2,f_2,g_2$ be the weights of the incoming and two outgoing edges of, respectively, $r_\mathcal{F}$ and $r_\mathcal{G}$. We write $\mathcal{F}_{0}$ and $\mathcal{F}_1$ for the sub-DDs of $\mathcal{F}$ with root nodes the 0- and 1-successors of $r_\mathcal{F}$ and weights $h_1\odot f_1$ and $h_1\odot g_1$, respectively. 
The two sub-DDs $\mathcal{G}_0$ and $\mathcal{G}_1$ of $\mathcal{G}$ are defined in a similar way. 
%Argue why $\mathcal{F}_i$ are reduced!
We assert that $\mathcal{F}_{i}$ and $\mathcal{G}_{i}$ are all reduced. Apparently, no reduction rule can be applied as they are sub-DDs of reduced symTDDs. Thus we need only show that they are normalised. By Theorem~\ref{thm:normalisation}, we need only show their root nodes are locally normalised. Take $\mathcal{F}_0$ as an example. Let $v=low(r_\mathcal{F})$ (the 0-successor of $r_\mathcal{F}$) and $f_v,g_v$ be the weights on the two outgoing edges of $v$ in $\mathcal{F}$. Note that $f_1$ is the weight on the incoming edge of $v$. We need only show $loc\_norm(h_1\odot f_1\odot f_v, h_1\odot f_1\odot g_v) = (h_1\odot f_1,f_v,g_v)$. Since $\mathcal{F}$ is normalised, we have $\supp{f_1}\subseteq \supp{h_1}$ and $loc\_norm(f_1\odot f_v,f_1\odot g_v)=(f_1,f_v,g_v)$. By Lemma~\ref{lemm:loc}, we have  $loc\_norm(h_1\odot f_1\odot f_v, h_1\odot f_1\odot g_v) = (h_1\odot f_1,f_v,g_v)$. Thus $\mathcal{F}_i$ and $\mathcal{G}_i$ are all normalised and reduced.
It is easy to see that $\mathcal{F}_0$ and $\mathcal{G}_0$ represent the same tensor $\phi|_{q=0}$ and hence are isomorphic by induction hypothesis. Analogously, $\mathcal{F}_1$ and $\mathcal{G}_1$ represent the same tensor $\phi|_{q=1}$ and are also isomorphic.

Let $\sigma_i$ be the isomorphism between $\mathcal{F}_i$ and $\mathcal{G}_i$ for $i=1,2$. Because $\mathcal{F}$ is reduced, for any node $v$ in both $\mathcal{F}_1$ and $\mathcal{F}_2$, we have $\sigma_1(v)=\sigma_2(v)$. Define $\sigma$ as the extension of $\sigma_1$ and $\sigma_2$ by further mapping $r_\mathcal{F}$ to $r_\mathcal{G}$. We show $\sigma:\mathcal{F}\to \mathcal{G}$ is an isomorphism. To this end, we need only show $h_1=h_2$, $f_1=f_2$, and $g_1=g_2$. Since $\mathcal{F}_i\cong\mathcal{G}_i$, we have $h_1\odot f_1=h_2\odot f_2$ and $h_1\odot g_1=h_2\odot g_2$. Because $\mathcal{F}$, $\mathcal{G}$ are normalised, we have  $(h_1,f_1,g_1)=loc\_norm(h_1\odot f_1,h_1\odot g_1)=loc\_norm(h_2\odot f_2,h_2\odot g_2)=(h_2,f_2,g_2)$. This shows $\mathcal{F}\cong  \mathcal{G}$.

\end{proof}

\section{Case Study}\label{sec:case}

In this section, we show by examples how symTDD can be used in the verification of quantum circuits. 

% The corresponding source code and results are available at \url{https://github.com/Veriqc/Symbol_TDD}.

%The first example is to do the symbolic execution of quantum Fourier transform. For this circuit, given a computational basis input, the output state can be represented in $\mathcal{O}(n)$ nodes using the original QMDD and TDD. But it needs $\mathcal{O}(2^n)$ nodes to represent the whole matrix of the circuit originally. 

\vspace*{2mm}
\noindent\textit{Quantum Fourier Transform.} Suppose the input state is a computational basis state. The output state of QFT can be represented as a QMDD or TDD with $\mathcal{O}(n)$ nodes. The functionality of the circuit, however, requires $\mathcal{O}(2^n)$ nodes to represent in either QMDD or TDD.

\begin{figure}[h]
    \centering
    \subfigure[]{
    \includegraphics[height=0.2\textwidth]{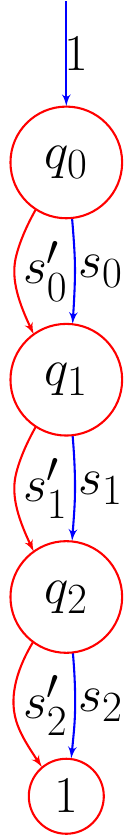}
    }
    \subfigure[]{
    \includegraphics[height=0.2\textwidth]{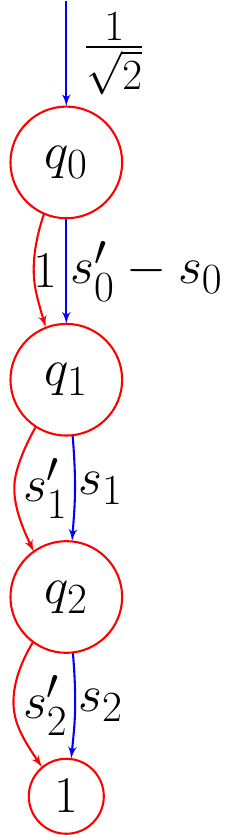}
    }
    \subfigure[]{
    \includegraphics[height=0.2\textwidth]{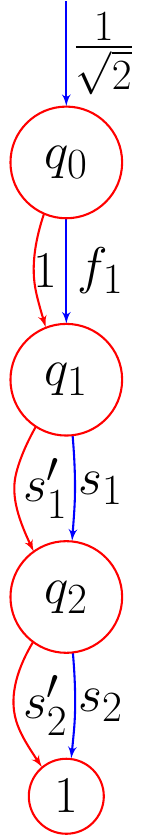}
    }
    \subfigure[]{
    \includegraphics[height=0.2\textwidth]{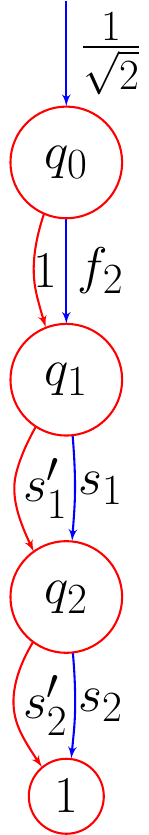}
    }
    \subfigure[]{
    \includegraphics[height=0.2\textwidth]{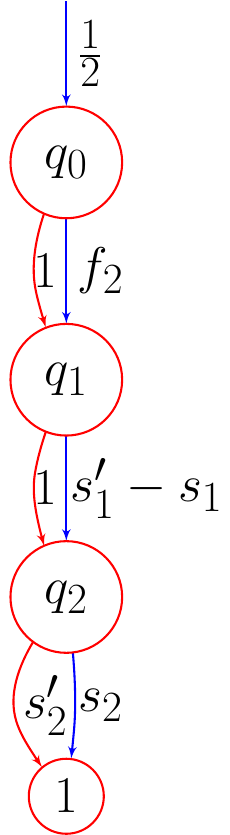}
    }
    \subfigure[]{
    \includegraphics[height=0.2\textwidth]{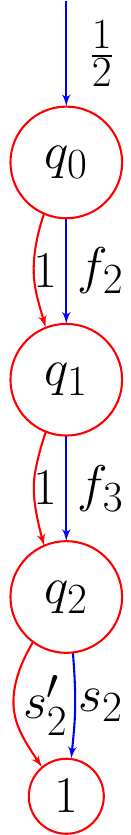}
    }
    \subfigure[]{
    \includegraphics[height=0.2\textwidth]{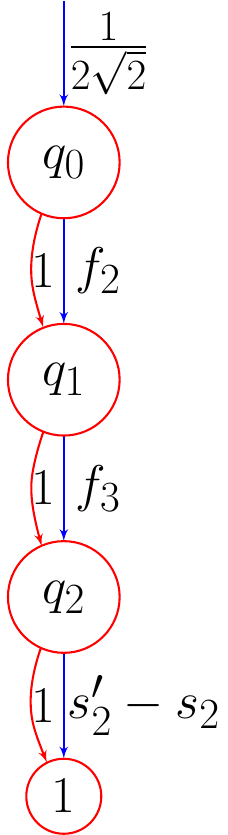}
    }
    \caption{symTDD for simulation of 3-qubit QFT, where (a)-(g) represent the symTDD of the input state and the states obtained after applying the gates shown in Fig. \ref{fig:exp-for-quantum-circuit} in sequence, and  $f_1=(s_0'-s_0)(s_1'-\imath s_1)$, $f_2=(s_0'-s_0 )(s_1'+\imath s_1)(s_2'+\frac{1+\imath}{\sqrt{2}}s_2 )$, $f_3=(s_1'-s_1 )(s_2'+\imath s_2 )$. }
    \label{fig:s_tdd_qft}
\end{figure}

To verify the functionality of the QFT circuit, we first execute the circuit on the symbolised input state $\ket{s_0}\cdots\ket{s_{n-1}}$. Fig.~\ref{fig:s_tdd_qft} illustrates the detailed process of executing the 3-qubit QFT on the symbolised input state $\ket{s_0}\ket{s_1}\ket{s_2}$. All symTDDs generated have a tower form with only four nodes. In this process, applying the Hadamard gate on $q_0$ changes the two weights $s_0^{\prime}$ and $s_0$ to $1$ and $s_0^{\prime}-s_0$ respectively. Then applying the Controlled-$R_2$ gate on $q_0,q_1$ further changes the weights to $1$ and $(s_0'-s_0)(s_1'-\imath s_1)$ respectively. This is because, no matter whether $R_2$ is applied on the target qubit $q_1$, the corresponding node is not changed (cf. Fig.~\ref{fig:exp_stdd}(a,c)) and only a weight is multiplied to the weight on the high-edge of the control qubit $q_0$. Similarly, applying the Controlled-$R_3$ gate on $q_0,q_2$ changes the weights to $1$ and $(s_0'-s_0 )(s_1'+\imath s_1)(s_2'+(\frac{1}{\sqrt{2}} + \frac{1}{\sqrt{2}}\imath)s_2 )$ respectively. The similar pattern will repeat on the last two qubits. The symTDD shown in Fig.~\ref{fig:s_tdd_qft}(g) contains the information of all possible output states of the circuit. Since $s_2'-s_2=e^{2\pi\imath0.s_2}$, $f_3=(s_1'-s_1 )(s_2'+\imath s_2 )=e^{2\pi\imath0.s_1s_2}$, and  $f_2=(s_0'-s_0 )(s_1'+\imath s_1)(s_2'+\frac{1+\imath}{\sqrt{2}}s_2 )=e^{2\pi\imath0.s_0s_1s_2}$, this verified the functionality of the QFT circuit.

\begin{figure}[h]
    \centering
     %\subfigure[]{
    \includegraphics[width=0.4\textwidth]{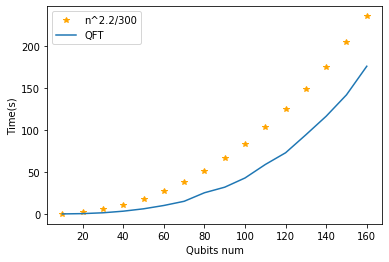}
     %}
     %\subfigure[]{
    \includegraphics[width=0.4\textwidth]{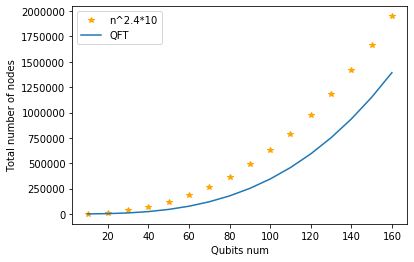}
     %}
    \caption{The time and space consumption of symTDD execution of 10-160 qubits QFT circuits. In this figure, we use \^\ to represent the exponent and n\^\ 2.2/300,  n\^\ 2.4*10 represent $n^{2.2}/300$ and $10 n^{2.4}$ respectively.}
    \label{fig:qft_data}
\end{figure}

% \delete{As a notable example, we further explore the performance of the symTDD on the quantum Fourier transform circuit.} 
We symbolically execute the QFT circuit from 10 to 160 qubits and depict in Fig.~\ref{fig:qft_data} the time and space consumption curve. The space consumption was measured as the total number of nodes used for representing both the symTDD and all tensor weights. As a comparison, we also give the curve for $n^{2.2}/300$ and $10n^{2.4}$ in the time and space figures. It can be seen that both the time and space consumption grow polynomially in the number of qubits. In fact, it only requires about ten thousand nodes to represent a 30-qubit QFT circuit in symTDD other than using around two billion nodes in a single TDD. Remarkably, QFT circuit with 160 qubits can be executed within three minutes on our laptop with 8GB memory.

% \red{SL: This discussion of QFT doesn't prove why the symTDD representation of $n$-qubit QFT has a tower form. A formal proof is necessary. Maybe we could show,  applying any gate in QFT and Toffoli circuits on any tower-form state, we will get another tower-form state?}

\vspace*{2mm}
\noindent\textit{Bernstein-Vazirani Algorithm.}
Suppose $O$ is an oracle such that $O\ket{\bold{x}}\ket{\bold{y}}=\ket{\bold{x}}\ket{\bold{y}\oplus f(\bold{x})}$, where $f(\bold{x})=\bold{s}\cdot \bold{x}=s_0x_0\oplus \cdots \oplus s_{n-1}x_{n-1}$ and $\bold{s}$ is a hidden Boolean string. The Bernstein-Vazirani algorithm (BV) \cite{bernstein1997quantum} can find $s$ with a single call of the oracle.
%\red{SL: please confirm.)}
% confirmed
The circuit for the 3-bit BV is shown in Fig.~\ref{fig:cir_bv}, where the oracle is implemented by a set of Controlled-$X^{s_i}$ gates. 
% \delete{For this circuit, the oracle is implemented by a set of Controlled-$X^{s_i}$ gates and the expected output state is $\ket{s_0}\cdots\ket{s_{n-1}}\ket{1}$.}

\begin{figure}[h]
    \centering
    \scalebox{0.8}{
\begin{quantikz}[column sep=0.4cm,row sep=0.28cm]
\lstick{$\ket{0}$}& \gate{H} &\ctrl{3}        &\qw            &\qw            & \gate{H} &\qw&\rstick{$\ket{s_0}$} \\
\lstick{$\ket{0}$}& \gate{H} &\qw             &\ctrl{2}       &\qw            & \gate{H} &\qw&\rstick{$\ket{s_1}$} \\
\lstick{$\ket{0}$}& \gate{H} &\qw             &\qw            &\ctrl{1}       & \gate{H} &\qw&\rstick{$\ket{s_2}$} \\
\lstick{$\ket{1}$}& \gate{H} &\gate{X^{s_0}}  &\gate{X^{s_1}} &\gate{X^{s_2}} & \gate{H} &\qw&\rstick{$\ket{1}$} \\
\end{quantikz}
}
\includegraphics[width=0.35\textwidth]{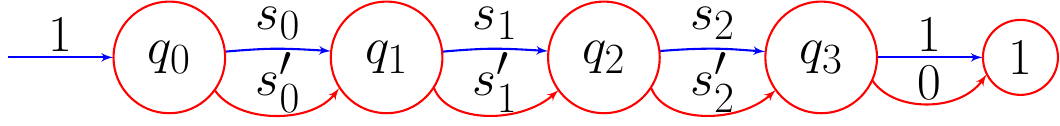}
\caption{The circuit (top) and symTDD (bottom) for verification of the Bernstein-Vazirani algorithm.}
\label{fig:cir_bv}
\end{figure}

% With symTDD, however, the Controlled-$X^{s}$ gate applied on the $k$-th and $l$-th qubits can be represented in a symbolic way, 
% where $X^{s}=\big[\begin{smallmatrix}s^{\prime}&s\\s&s^{\prime}\end{smallmatrix}\big]$ \yf{this sentence doesn't make much sense. Explain it in more detail or simply delete it}.

While for every $\bold{s}$ this circuit can be verified with a single run of TDD simulation, it needs exponential time in total to verify this for all $\bold{s}$. Executing the circuit symbolically, we obtain the symTDD of $\ket{s_0}\cdots\ket{s_{n-1}}\ket{1}$ as the output state (cf. Fig.~\ref{fig:cir_bv} (bottom)), indicating that the hidden string $\bold{s}$ has been successfully computed. It is also observed in our experiments that all symTDDs generated during this process are in a tower form.
% , since that Controlled-$X^{s_i}$ together with the two Hadamard gates and the input state $\ket{1}$ on the last qubit make it equivalent to apply a $Z^{s_i}=\begin{bmatrix}1&0\\0&s_i^{\prime}-s_i\end{bmatrix}$ gate on the $i$-th qubit.

% \begin{figure}[h]
%     \centering
    % \subfigure[]{
    % \includegraphics[height=0.25\textwidth]{figures/cx.pdf}
    % }
    % \subfigure[]{
    % \includegraphics[height=0.25\textwidth]{figures/bv.pdf}
    % }
    % \subfigure[]{
    % \includegraphics[height=0.2\textwidth]{figures/grover.pdf}
    % }
    % \includegraphics[width=0.35\textwidth]{figures/bv_rotate.pdf}
    % \caption{Symbolic verification of BV algorithm.
    % \red{(SL: Remove (a),(b) from the figure to save space. Explain more about these two DDs.)}
%     }
%     \label{fig:tdd_bv}
% \end{figure}

\begin{figure}[h]
    \centering
    % \subfigure[]{
    \includegraphics[width=0.4\textwidth]{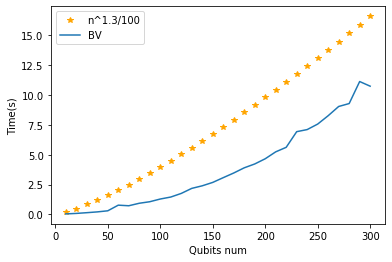}
    % }
    % \subfigure[]{
    \includegraphics[width=0.4\textwidth]{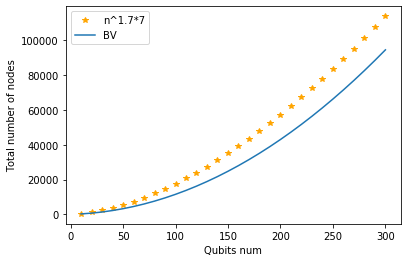}
    % }
    \caption{The time and space consumption of using symTDDs to verify the Bernstein-Vazirani algorithm. In the figure, we use \^\ to represent the exponent and n\^\ 1.3/100,  n\^\ 1.7*7 represent $n^{1.3}/100$ and $7n^{1.7}$ respectively.
    }
    \label{fig:bv_data}
\end{figure}

Fig. \ref{fig:bv_data} gives the time and space consumption needed for verifying BV using the symTDD. As a comparison, we also provide the data for $n^{1.3}/100$ and $7n^{1.7}$. It can be seen that both the time and space consumption grow polynomially in the number of qubits. In fact, 300-qubit BV can be verified within 10 seconds.

% \red{SL: Again, this does not prove the result for general BV. Does it hold that applying Controlled-$X^{s}$ to a tower-form state always yield a tower-form state?}

\vspace*{2mm}
\noindent\textit{Grover's Algorithm.} 
Similar to BV, Grover's algorithm \cite{nielsen2002quantum} can also be verified by symTDD. For the 3-qubit Grover's algorithm, an oracle $O$ as specified by $O\ket{\bold{x}}\ket{y}=\ket{\bold{x}}\ket{y\oplus f(\bold{x})}$ with $f(\bold{x})=1$ iff $x=s$, where $s$ is the solution that the algorithm aims to find.
% \yf{Are you assuming that $s = s_1s_0$ is the only solution of problem? If yes, you should state it clearly!} \yf{Now if so, why $f(x)=x_0^{s_0^\prime}x_1^{s_1^\prime}$ gives the right oracle? Explain it!} \yf{Note that if $x_0=x_1=1$, then $x_0^{s_0^\prime}x_1^{s_1^\prime}=1$ no matter what $s_0$ and $s_1$ are. That is, 3 is always one of the solutions no matter what $s$ is!}. 
Similar to BV, it also needs an exponential number of runs to verify for all possible $\bold{s}$ by using TDD. In contrast, a single run of the symTDD suffices. 

% Here the symTDD representation of the output is shown in Fig. \ref{fig:tdd_bv} (c), which exactly represents the state $\ket{s_0}\ket{s_1}\ket{1}$ up to a global phase $-1$.

\begin{figure}[h]
\centering
\scalebox{0.72}{
\begin{quantikz}[column sep=0.28cm,row sep=0.4cm]
\lstick{$\ket{0}$}&\gate{H}&\gate{X^{s_0^{\prime}}}&\ctrl{1}&\gate{X^{s_0^{\prime}}}&\gate{H}&\gate{X}&\qw     &\ctrl{1}&\qw     &\gate{X}&\gate{H}&\qw&\rstick{$\ket{s_0}$} \\
\lstick{$\ket{0}$}&\gate{H}&\gate{X^{s_1^{\prime}}}&\ctrl{1}&\gate{X^{s_1^{\prime}}}&\gate{H}&\gate{X}&\gate{H}&\gate{X}&\gate{H}&\gate{X}&\gate{H}&\qw&\rstick{$\ket{s_1}$} \\
\lstick{$\ket{1}$}&\gate{H}&\qw           &\gate{X}&\qw           &\qw     &\qw     &\qw     &\qw     &\qw     &\qw     &\gate{H}&\qw&\rstick{$\ket{1}$} \\
\end{quantikz}
}
\includegraphics[width=0.3\textwidth]{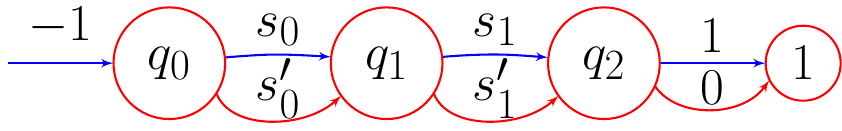}
\caption{The circuit (top) and symTDD (bottom) for verification of the 3-qubit Grover's algorithm.}
\label{fig:cir_grover}
\end{figure}

% Here the symTDD representation of the $X^{s^{\prime}}$ gate applied on the $k$-th qubit is shown in Fig. \ref{fig:tdd_grover} (a), and 
%The output of symbolically simulation of the circuit is shown in Fig.~ \ref{fig:cir_grover} (bottom), which exactly represents the state $\ket{s_0}\ket{s_1}\ket{1}$ up to a global phase $-1$.
% \yf{it might be better to trace out the last qubit}

% \begin{figure}[h]
%     \centering
%     \begin{tabular}{ccc}
%     \includegraphics[height=0.2\textwidth]{figures/Xs.pdf}
%  & \quad\quad\quad &
%     \includegraphics[height=0.2\textwidth]{figures/grover.pdf}
%     \end{tabular}
%    Here, $f_0=-0.177(s_2+s_2^{\prime}s_1+s_2^{\prime}s_1^{\prime}s_0)-0.884s_2^{\prime}s_1^{\prime}s_0^{\prime}$, $f_1=s_2+s_2^{\prime}s_1+5s_2^{\prime}s_1^{\prime}s_0+0.2s_2^{\prime}s_1^{\prime}s_0^{\prime}$, $f_2=s_2+s_2^{\prime}s_0+5s_2^{\prime}s_1s_0^{\prime}+0.2s_2^{\prime}s_1^{\prime}s_0^{\prime}$, $f_3=s_2+5s_2^{\prime}s_1s_0+s_2^{\prime}s_0^{\prime}+0.2s_2^{\prime}s_1^{\prime}s_0$,$f_4=s_1+s_1^{\prime}s_0+5s_2s_1^{\prime}s_0^{\prime}+0.2s_2^{\prime}s_1^{\prime}s_0^{\prime}$,$f_5=s_1^{\prime}+s_1s_0+5s_2s_1^{\prime}s_0+0.2s_2^{\prime}s_1s_0^{\prime}$, $f_6=s_1+s_1^{\prime}s_0^{\prime}+5s_2s_1^{\prime}s_0+0.2s_2^{\prime}s_1^{\prime}s_0$, $f_7=s_1^{\prime}+s_1s_0^{\prime}+5s_2s_1s_0+0.2s_2^{\prime}s_1s_0$.
%     \includegraphics[width=0.35\textwidth]{figures/grover_rotate.pdf}
%     \caption{symTDDs for verification of Grover's algorithm.}
%     \label{fig:tdd_grover}
% \end{figure}

The output of symbolic execution of the 3-qubit Grover circuit is shown in Fig.~ \ref{fig:cir_grover} (bottom), which exactly represents the state $\ket{s_0}\ket{s_1}\ket{1}$ up to a global phase $-1$.

For Grover's algorithm with more than 4 qubits, the output state is not exactly $\ket{s_0}\cdots \ket{s_{n-1}}\ket{1}$. Nevertheless, it can be analysed using symTDD as well. We symbolically execute the algorithm by applying the Grover iteration $\lfloor \sqrt{2^n}\pi/4 \rfloor$ times, and then compute the probability of successfully finding the solution $s$. Note that the desirable output $\ket{s_0}\cdots \ket{s_{n-1}}\ket{1}$ can also be represented as a symTDD $\phi_{\mathit{suc}}$. The calculation of the success probability 
boils down to the contraction of $\phi_{\mathit{suc}}$ with the resulting symTDD of executing Grover's algorithm. 
% \delete{We add the output  at the end of the circuit, and then simulate the circuit symbolically. The resulted symTDD has only one node with a weight which is a constant tensor with real number representing the amplitude of $\ket{s_0}\cdots \ket{s_{n-1}}\ket{1}$, and the square of it is the probability.}

Although the whole complexity is not polynomial, it shows an exponential advantage over TDDs. For example, it takes 29 seconds and 53,892 nodes to simulate the 8-qubit Grover's algorithm (with 8 iterations) using symTDD and find the success probability to be 0.9956 no matter what the solution is. In comparison, it takes 23 seconds and 1753 nodes to obtain the same answer for a given solution, which means that the expenses will sum up to $23 \cdot 2^7=2,944$ seconds and $1,753\cdot 2^7=224,384$ nodes when the task is to analyse all possible $\bold{s}$. When we consider the 9-qubit Grover's algorithm with 12 iterations, the corresponding data will be $153\ vs\ 34,371$ seconds and $164,377\ vs\ 829,184 $ nodes respectively, where the success probability is 0.9999.

\vspace*{2mm}
\noindent\textit{Toffoli Circuits.}
In addition, symTDD can be used to analyse Toffoli circuits \cite{szyprowski2011reducing}, i.e., circuits consisting of only $C^kX$ ($k\in \mathbb{N}$) gates. Every Toffoli circuit has a symTDD representation in tower form (cf. the bottom of Fig.~\ref{fig:cir_bit_flip}). This is because, for any $C^k X$ which splits the symTDD, the two branches have disjoint supports, and thus can be merged by using RR3 introduced in Sec.~\ref{sec:red}. 
% $C^kX\ket{f_0}\cdots\ket{f_{k-2}}\ket{f_{k-1}}=\ket{f_0}\cdots\ket{f_{k-2}}\ket{f_0\wedge \cdots \wedge f_{k-2}\oplus f_{k-1}}$, where $f_0, \cdots, f_{k-1}$ are all Boolean functions. Thus, applying a $C^kX$ on $k$ qubits is equivalent to apply a $X^{f}=\begin{bmatrix}f^{\prime}&f\\f&f^{\prime}\end{bmatrix}$ gate to the last qubit, where $f=f_0\wedge \cdots \wedge f_{k-2}$.
% \yf{Can't Toffoli circuits be efficiently simulated by ordinary OBDD?}
%
% Thus, for such a circuit, the complexity for simulating it using symTDD is $\mathcal{O}(n\cdot |\mathcal{B}|)$, where $|\mathcal{B}|$ is the biggest number of nodes used for representing the complex-valued Boolean functions and $n$ is the number of qubits. 
%
Moreover, the two weights on the outgoing edges of every internal node have forms $f^\prime$ and $f$ respectively, where $f$ is a tensor that states exactly the relations between the output state of this qubit and the input signals while $f^{\prime}$ represents its complement. By contrast, using TDD will interleave the input and output indices, making the analysis difficult.
\begin{figure}[h]
    \centering
%    \subfigure[]{%
    \scalebox{0.8}{%
    \begin{quantikz}[column sep=0.28cm,row sep=0.15cm]%
    \lstick{$\ket{s_0}$}&&\ctrl{3}&\qw     &\ctrl{5}&\qw     &\qw     &\qw     &\gate{X} &\qw      &\qw      &\qw\\
    \lstick{$\ket{s_1}$}&&\qw     &\ctrl{2}&\qw     &\ctrl{3}&\qw     &\qw     &\qw      &\gate{X} &\qw      &\qw\\
    \lstick{$\ket{s_2}$}&&\qw     &\qw     &\qw     &\qw     &\ctrl{2}&\ctrl{3}&\qw      &\qw      &\gate{X} &\qw\\
    \lstick{$\ket{0}$}&&\gate{X}&\gate{X}&\qw     &\qw     &\qw     &\qw     &\ctrl{-3}&\ctrl{-2}&\qw      &\meter{}\\
    \lstick{$\ket{0}$}&&\qw     &\qw     &\qw     &\gate{X}&\gate{X}&\qw     &\qw      &\ctrl{-1}&\ctrl{-2}&\meter{}\\
    \lstick{$\ket{0}$}&&\qw     &\qw     &\gate{X}&\qw     &\qw     &\gate{X}&\ctrl{-2}&\qw      &\ctrl{-1}&\meter{}
    \end{quantikz}
    }%
    % \includegraphics[height=0.15\textwidth]{figures/cir_bit_flip.pdf}
    % \label{fig:cir_bit_flip}
    \\
    \includegraphics[width=0.4\textwidth]{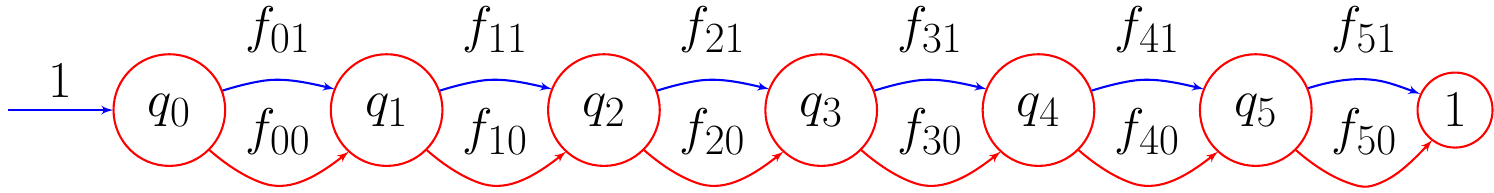}    
    % \label{fig:tdd_bit_flip}
    
    \caption{The circuit (top) and the symTDD (bottom) for verification of the bit-flip code circuit. Principle of deferred measurement \cite{nielsen2002quantum} is used.
    \label{fig:cir_bit_flip}
    } 
%    }%
%    \quad
%    \subfigure[]{%
%        \includegraphics[height=0.335\textwidth]{figures/bit_flip.pdf}
%        \label{fig:tdd_bit_flip}
%        }%
%    \\
%    \subfigure[]{%
%        \includegraphics[width=0.45\textwidth]{figures/bit_flip-li}
%        \label{fig:tdd_bit_flip}
%        }%
%    \caption{(a) The circuit for bit-flip code. Principle of deferred measurement \cite{nielsen2002quantum} is used. (b) symTDD for verification of bit-flip code circuit. \red{(SL: The DD is too small. Consider redrawing it.)}}
    % \label{fig:cir_grover}
\end{figure}

% \begin{figure}[h]
%     \centering
%     \includegraphics[height=0.3\textwidth]{figures/bit_flip.pdf}
%     \caption{symTDD for verification of bit-flip code circuit. \magenta{Try to combine this Fig with Fig. \ref{fig:cir_grover}?}}
%     \label{fig:tdd_bit_flip}
% \end{figure}

% \yf{you should at least explain what the functionality of this circuit is}

As an example of Toffoli circuits, let us consider the bit-flip code circuit \cite{nielsen2002quantum} shown in Fig.~\ref{fig:cir_bit_flip} (top), which can correct the input quantum state if there is at most one bit-flip error occurred. Let $q_0,\cdots , q_{5}$ be the 6 qubits from top to bottom. Symbolically executing this circuit, the corresponding symTDD is shown in Fig.~\ref{fig:cir_bit_flip} (bottom) where $f_{00}=f_{10}=f_{20}=s_0^{\prime}s_1^{\prime}s_2+s_0^{\prime}s_1s_2^{\prime}+s_1^{\prime}s_2^{\prime}$ and $f_{01}=f_{11}=f_{21}=f_{00}^{\prime}$. This means that the output of $q_0,q_1,q_2$ will be $\ket{000}$ when at least two of $s_0,s_1,s_2$ are 0 and the output of $q_0,q_1,q_2$ will be $\ket{111}$ when at least two of $s_0,s_1,s_2$ are 1. Furthermore, the weights $f_{30}=f_{31}^{\prime}=s_0s_1+s_0^{\prime}s_1^{\prime}$, $f_{40}=f_{41}^{\prime}=s_1s_2+s_1^{\prime}s_2^{\prime}$ and $f_{50}=f_{51}^{\prime}=s_0s_2+s_0^{\prime}s_2^{\prime}$, meaning that measuring $q_3,q_4,q_5$ at the end of the circuit will obtain result 0 iff $s_0=s_1$, $s_1=s_2$ and $s_0=s_2$, respectively. Note that to obtain the same information, we have to traverse all 56 paths in a TDD of 34 nodes, if TDD is employed for the verification task.

\section{Conclusion}\label{sec:conclusion}

This paper proposed the first decision-diagram approach for operating symbolic objects. The proposed symTDD is an extension of the tensor decision diagram and makes it possible to leverage the power of symbolic logic and tensor networks in a systematic way. Our experiments on QFT, BV, Grover's Algorithm have partially demonstrated the utility of symTDD. In particular, it provides an efficient approach for extracting information of the output states of quantum circuits. It can also be used to verify the correctness of quantum algorithms with classical control signals. Future work will employ SMT solvers like Z3 in extracting more useful output state information of quantum circuits. 

%and provide a potential for exponential speedup compared to the original TDD. Our method gives the first practical and usable tool for using the symbolic method to perform the verification of parametric Boolean quantum circuits. It shall give people more inspiration to inspect quantum algorithms in a symbolic way. 

\section*{Acknowledgements}

This work was partially supported by the Australian Research Council (Grant No.: DP220102059). The research of Xin Hong was also partially supported by Sydney Quantum Academy.

\bibliographystyle{IEEEtran}

\bibliography{references}

% \subsection*{Appendix: The step-by-step run of QFT}\label{sec:appendix}

% \begin{figure}[h]
%     \centering
%     \subfigure[]{
%     \includegraphics[height=0.2\textwidth]{figures/qft1.pdf}
%     }
%     \subfigure[]{
%     \includegraphics[height=0.2\textwidth]{figures/qft2.pdf}
%     }
%     \subfigure[]{
%     \includegraphics[height=0.2\textwidth]{figures/qft3.pdf}
%     }
%     \subfigure[]{
%     \includegraphics[height=0.2\textwidth]{figures/qft4.pdf}
%     }
%     \subfigure[]{
%     \includegraphics[height=0.2\textwidth]{figures/qft5.pdf}
%     }
%     \subfigure[]{
%     \includegraphics[height=0.2\textwidth]{figures/qft6.pdf}
%     }
%     \subfigure[]{
%     \includegraphics[height=0.2\textwidth]{figures/qft7.pdf}
%     }
%     \caption{symTDD for verification of QFT. Here, $f_1=(s_0'-s_0)(s_1'-\imath s_1)$, $f_2=(s_0'-s_0 )(s_1'+\imath s_1)(s_2'+(\frac{1}{\sqrt{2}} + \frac{1}{\sqrt{2}}\imath)s_2 )$, $f_3=(s_1'-s_1 )(s_2'+\imath s_2 )$. }
%     \label{fig:s_tdd_qft}
% \end{figure}

%\subsection*{Appendix: A further reduction rule}\label{sec:appendix}

% But since it will change the weights, we will use it in the calculation process but not use it in making a canonical representation.

% $\imath$
%输入是Boolean function
%Reachability
%图覆盖，VQA，图分割
% Error correction, surface code, graph state, distributed quantum computation

\end{document}